\documentclass[review]{elsarticle}
\usepackage{lineno}
\modulolinenumbers[5]
\makeatletter
\def\ps@pprintTitle{%
 \let\@oddhead\@empty
 \let\@evenhead\@empty
 \def\@oddfoot{}%
 \let\@evenfoot\@oddfoot}
\makeatother

\usepackage[T1]{fontenc}
\usepackage[latin9]{inputenc}
\usepackage{geometry}
\usepackage{mathtools}
\usepackage{amsmath}
\usepackage{amsthm}
\usepackage{amssymb}
\usepackage{stackrel}

\theoremstyle{definition}
\newtheorem{defn}{\protect\definitionname}
\theoremstyle{plain}
\newtheorem{prop}{\protect\propositionname}
\newtheorem{cor}{\protect\corollaryname}
\newtheorem{lem}{\protect\lemmaname}
\newtheorem{thm}{\protect\theoremname}
\newtheorem{rem}{\protect\remarkname}
\theoremstyle{definition}
\newtheorem{example}{\protect\examplename}
\theoremstyle{remark}
\newtheorem{claim}{\protect\claimname}

\usepackage{listings}

\setcitestyle{round}
\usepackage[hang,flushmargin]{footmisc}
\usepackage[charter]{mathdesign}

\usepackage[english]{babel}
\addto\captionsaustralian{\renewcommand{\claimname}{Claim}}
\addto\captionsaustralian{\renewcommand{\corollaryname}{Corollary}}
\addto\captionsaustralian{\renewcommand{\definitionname}{Definition}}
\addto\captionsaustralian{\renewcommand{\examplename}{Example}}
\addto\captionsaustralian{\renewcommand{\lemmaname}{Lemma}}
\addto\captionsaustralian{\renewcommand{\propositionname}{Proposition}}
\addto\captionsaustralian{\renewcommand{\remarkname}{Remark}}
\addto\captionsaustralian{\renewcommand{\theoremname}{Theorem}}
\addto\captionsenglish{\renewcommand{\claimname}{Claim}}
\addto\captionsenglish{\renewcommand{\corollaryname}{Corollary}}
\addto\captionsenglish{\renewcommand{\definitionname}{Definition}}
\addto\captionsenglish{\renewcommand{\examplename}{Example}}
\addto\captionsenglish{\renewcommand{\lemmaname}{Lemma}}
\addto\captionsenglish{\renewcommand{\propositionname}{Proposition}}
\addto\captionsenglish{\renewcommand{\remarkname}{Remark}}
\addto\captionsenglish{\renewcommand{\theoremname}{Theorem}}

\providecommand{\claimname}{Claim}
\providecommand{\corollaryname}{Corollary}
\providecommand{\definitionname}{Definition}
\providecommand{\examplename}{Example}
\providecommand{\lemmaname}{Lemma}
\providecommand{\propositionname}{Proposition}
\providecommand{\remarkname}{Remark}
\providecommand{\theoremname}{Theorem}

\begin{document}
\begin{frontmatter}

\title{On comparison of expert}
\author[mymainaddress]{Itay Kavaler\corref{mycorrespondingauthor}}
\cortext[mycorrespondingauthor]{Corresponding author}
\ead{itayk@campus.technion.ac.il} 
\author[mymainaddress]{Rann Smorodinsky} 
\ead{rann@ie.technion.ac.il}
\address[mymainaddress]{Davidson Faculty of Industrial Engineering and Management, Technion, Haifa 3200003, Israel}

\begin{abstract}
A policy maker faces a sequence of unknown outcomes. At each stage two (self-proclaimed) experts provide probabilistic forecasts on the outcome in the next stage. A comparison test is a protocol for the policy maker to (eventually) decide which of the two experts is better informed. The protocol takes as input the sequence of pairs of forecasts and actual outcomes and (weakly) ranks the two experts.

We focus on anonymous and non-counterfactual comparison tests and propose two natural properties to which such a comparison test must adhere. We show that these determine the test in an essentially unique way. The resulting test is a function of the derivative of the induced pair of measures at the realized outcomes. 
\end{abstract}

\begin{keyword}
\texttt{forecasting\sep probability\sep testing} 
\JEL\texttt{C11\sep C70\sep C73\sep D83}
\end{keyword}

\end{frontmatter}

\section{{\normalsize{}Introduction }}

{}{} The literature on expert testing has, by and large, treated
the question of whether a self-proclaimed expert can be identified
as such, while also not allowing for charlatans to pass the test.
A striking result due to \Citet*{Sandroni-2003} is that no such test
exists without additional structural assumptions regarding the problem.
The basic premise of this literature is the validity of the underlying
question of whether a forecaster, or rather a probabilistic model,
is correct or false. In a hypothetical world, where only one model
exists and the tester can only entertain the services of a single
expert, this may make sense. Even then, one might wonder what is the
tester to do whenever she rejects the expert. Does she turn to another
expert or to her own intuition? In any case she would probably, implicitly,
utilize an alternative (possibly untested) model.

This motivates us to seek an alternative approach to the issue
of expert testing and that is a comparison of experts, which is the
approach we pursue here. In this approach the tester is exposed to
a few alternative models (forecasters) and a single realization of
events. The tester then compares the alternative forecasters and decides
which is the better informed one. Facing many (possibly conflicting)
experts is commonplace in weather forecasting, financial forecasting,
medical prognosis and more. Nevertheless, the design of comparison
tests has been almost entirely ignored in the literature on expert
testing. Two exceptions are \Citet*{Al-Najjar-2008} and \Citet*{Feinberg-Stewart-2008}
which we discuss in the section on related literature.

{}The approach we take in this paper is axiomatic. After defining
exactly what is meant by a comparison test we will turn to discuss
some desirable properties for such tests. We then construct a test
that complies with all the desired properties and show it is essentially
unique. The setting we focus on is that of two experts and a test
which (weakly) ranks the two and hence its range consists of three
outcomes. It may either point at one of the two experts as being better
informed or it may be indecisive. Let us discuss the properties that
are central to our main results.

\textbf{Anonymity} - A test is {\em anonymous} if it does not depend
on the identity of the agents but only on their forecasts.

\textbf{Error-free} - Let us assume that one of the experts has the
correct model (namely, he would have passed a standard single expert
test which has no type-1 errors). An {\em error-free} test will
surely not point at the second expert as the superior one (albeit,
it may provide a non-conclusive outcome).

\textbf{Reasonable }- Let us consider an event, $A$, that has positive
probability according to the first expert but zero probability according
to the second. Conditional on the occurrence of the event $A,$ a
{\em reasonable} test must assign positive probability to the first
expert being better informed than the second.

The approach taken in this paper can be considered as a contribution
to the hypothesis testing literature in statistics where a forecaster
is associated with a hypothesis. In this context we propose a hypothesis
test that complies with a set of fundamental properties which we refer
to as axioms. In contrast, a central thrust for the hypothesis testing
literature (for two hypotheses) is the pair of notions of significance
level and power of a test. In that literature one hypothesis is considered
as the null hypothesis while the other serves as an alternative. A
test is designed to either reject the null hypothesis, in which case
it accepts the alternative, or fail to reject it (a binary outcome).
The significance level of a test is the probability of rejecting the
null hypothesis whenever it is correct (type-1 error) while the power
of the test is the probability of rejecting the null hypothesis assuming
the alternative one is correct (the complement of a type-2 error). 

In contrast with the aforementioned binary outcome that is prevalent
in the hypothesis testing literature we allow, in addition, for an
inconclusive outcome. Recall the celebrated Neyman-Pearson lemma which
characterizes a test with the maximal power subject to an upper bound
on the significance level. The possibility of an inconclusive outcome,
in our framework, allows us to design a test where both type-1 and
type-2 errors have zero probability.\footnote{Note that we abuse the statistical terminology. In statistics the
notion of rejection is always used in the context of the null hypothesis.
In our model we assume symmetry between the alternatives and so we
discuss rejection also in the context of the alternative hypothesis.
As a consequence, an error of type-1 is defined as the probability
of accepting the alternative hypothesis whenever the null hypothesis
is correct, and symmetrically, an error of type-2 is the probability
of accepting the null hypothesis whenever the alternative one is correct.}

Interestingly, the test proposed in the Neyman-Pearson lemma, similar
to ours, also hinges on the likelihood ratio.\footnote{The test proposed in the Neyman-Pearson lemma rejects the null hypothesis
whenever the likelihood ratio falls below some positive threshold.} In our approach we, a priori, treat both hypotheses symmetrically.
In the statistics literature, however, this is not the case and the
null hypothesis is, in some sense, the status quo hypothesis. This
asymmetry is manifested, for example, in the Neyman-Pearson lemma.

Note that in order to design a test that complies with a given significance
level and a given power one must know the full specification of the
two hypotheses. This is in contrast with our test which is universal,
in the sense that it does not rely on the specifications of the two
forecasts. Finally, let us comment that whereas hypothesis testing
is primarily discussed in the context of a finite sample, typically
from some IID distribution, our framework allows for sequences of
forecasts that are dependent on past outcomes as well as past forecasts
of the other expert.

\subsection{{\normalsize{}Results}}

We construct a specific comparison test based on the derivative of
two measures that are induced by the two forecasters. We prove that
this test is anonymous, error-free and reasonable. 

Two tests are essentially equal if their verdict is equal with probability
one for any pair of forecasters.\footnote{``with probability one\textquotedbl{} is meant with respect to the
probability measure induced by either of the two forecasters. Please
refer to Definition \ref{Def: different test } for a more rigorous
statement.} The test we construct turns out to be unique modulo this equivalence
relation. In other words, for any test that is not equivalent to ours
and is anonymous and reasonable there exist two forecasters for which
an error will be made (the probability of reversing the order) and
hence that test cannot be error-free.

Finally, our constructed test perfectly identifies the correct forecaster
whenever the two measures induced by the forecasters are mutually
singular with respect to each other. Requiring the test to identify
the correct expert when the measures are not mutually singular is
shown to be impossible.

\subsection{{\normalsize{}\label{subsec:Related-Literature}Related literature }}

Much of the literature on expert testing focuses on the single expert
setting. This literature dates back to the seminal paper of \Citet*{Dawid-1982}
who proposes the calibration test as a means to evaluate a forecaster
(in particular a weather forecaster) and shows that a true expert
will never fail this test. \Citet*{Foster-and-Vohra-1998} show how
a charlatan, who has no knowledge of the weather, can produce forecasts
which are always calibrated. The basic ingredient that allows the
charlatan to fool the test is the use of random forecasts. \Citet*{Lehrer-2001}
and \Citet*{Sandroni-and-Smorodinsky-and-Vohra-2003} extend this
observation to a broader class of calibration-like tests. {}Finally,
\Citet*{Sandroni-2003} shows that there exists no error-free test
that is immune to such random charlatans (see also extensions of Sandroni's
result in \citet{Shmaya-2008} and \Citet*{Olszewski-and-Sandroni-2008}).

To circumvent the negative results various authors suggest to limit
the set of models for which the test must be error-free (e.g., \Citet*{Al-Najjar-2010},
and \Citet*{Pomatto-2016}), or to limit the computational power associated
with the charlatan (e.g., \Citet*{Fortnow-and-Vohra-2009}) or to
replace measure-theoretic implausibility with topological implausibility
by resorting to the notion of \textit{{}category one sets}{} (e.g.,
\Citet*{Dekel-Feinberg}).

As previously mentioned, the comparison of experts has drawn little
attention in the community studying expert testing, with two exceptions.
\Citet*{Al-Najjar-2008} proposed a test based on the likelihood ratio
for comparing two experts. They show that if one expert knows the
true process whereas the other is uninformed, then one of the following
must occur: either, the test correctly identifies the informed expert,
or the forecasts made by the uninformed expert are close to those
made by the informed one. It turns out that the test they propose
is anonymous and reasonable but is not error-free (Subsection \ref{subsec:The-likelihood-ratio},
Claim \ref{claim 1: al-najar}). An asymptotic version of this likelihood
ratio, however, will play a crucial role in our construction.

\Citet*{Feinberg-Stewart-2008} study an infinite horizon model of
testing multiple experts using a cross-calibration test. In their
test $N$ experts are tested simultaneously; each expert is tested
according to a calibration restricted to dates where not only does
the expert have a fixed forecast but the other experts also have a
fixed forecast, possibly with different values (a formal definition
is given in Appendix \ref{sec:Appendix B The-Cross-Calibration}).
They showed that a true expert is guaranteed to pass the cross-calibration
no matter what strategies are employed by the other experts.

In addition, they prove that in the presence of an informed expert,
the subset of data-generating processes under which an ignorant expert
(a charlatan) will pass the cross-calibration test with positive probability,
is topologically ``small''. The cross calibration test naturally
induces a comparison test for two experts: If one expert passes while
the other does not then he is the better informed one, while in all
other cases the test is inconclusive.\footnote{In fact, any single expert test induces a comparison test as follows.
Run the test for each of the two experts simultaneously and whenever
one passes and the other one fails rank them accordingly. Otherwise,
the test is inconclusive.} This induced comparison test turns out to be anonymous and error-free
but not reasonable (for further details see Claim \ref{claim 2 cross calibration is not reasonable}
in Subsection \ref{subsec:The-Cross-Calibration-test}).

\citet{Echenique-Shmaya-2008} study a setting where a decision maker
(DM) has some initial belief about the evolution of a system and takes
actions to maximize her payoff. The DM is offered an alternative hypothesis
and the paper provides a scheme for choosing between the two hypotheses
(a `test') with a guarantee on the payoffs. In particular, whenever
the scheme suggests to adopt the alternative hypothesis, the resulting
payoffs do not diminish in comparison with the hypothetical payoff
were that hypothesis rejected. In addition, their test is shown to
accept the initial belief whenever it is true. Their test, once again,
is based on the likelihood ratio but is obviously asymmetric and is
not error-free.

\Citet*{Pomatto-2016} poses a question that can be interpreted as
one about multiple expert testing. The paper characterizes classes
of hypotheses (`paradigms' in his jargon) for which there exists a
test that will pass the true hypothesis while rejecting any other
hypothesis in the class as well as any convex combination thereof.
The latter requirement is quite strong as it consequently means that
any pair of hypotheses is not testable, in contrast with our results.

Finally, the likelihood ratio, central to our result, appears in the
context of many statistical tests. Whereas our work derives a test
based on the likelihood ratio as an essentially unique test that conforms
with some fundamental properties, many papers and scholars in statistics
consider the likelihood ratio as axiomatic. This is captured in \citet{Edwards}'s
well-cited Likelihood Axiom: ``Within the framework of a statistical
model, all the information which the data provide concerning the relative
merits of two hypotheses is contained in the likelihood ratio of those
hypotheses on the data, and the likelihood ratio is to be interpreted
as the degree to which the data support the one hypothesis against
the other''.

\subsection{{\normalsize{}Finite or infinite test?}}

A long-standing debate in the literature on expert testing is whether
a test should be finite. A test is finite if its decision is made
in some finite time. In contrast, an infinite test may require the
infinite sequence of forecasts and realizations prior to making a
verdict. The argument for considering finite tests is that infinite
tests are impractical.

Although we sympathize with the argument that infinite tests are impractical
we do think they have academic merit. The construction of well-behaved
infinite, possibly impractical, tests would eventually shed light
on their finite counterpart. Thus, if the technical analysis underlying
the understanding of infinite tests is more tractable than that of
finite tests, then the study of infinite tests should be the port
of embarkation for this research endeavor. This is what motivates
our approach in this paper.\footnote{In a companion paper (\citet{Kavaler-Smorodinsky-2019}) we use some
of the machinery developed here to formulate a `well-behaved' finite
test.}

Furthermore, in expert testing we should allow experts to calibrate
their model given the data. Pushing the design of tests towards finite
tests may result in tests that give a verdict before these models
are refined and calibrated. Consider the classical example of an IID
process. A forecaster who is aware that indeed the process is such
may need time (and data) to calibrate the model and to learn its parameter.
Initial forecasts may be wrong, yet those made after a calibration
phase become more accurate and long-run predictions are spot-on. To
capture the importance of such a preliminary calibration test and
patience in model (expert) selection we introduce the following notion:\footnote{In a way the recent success of `deep learning' based on enormous data
sets (paralleling our interest in long-run observations) testifies
to the importance of patience in model (expert) selection and the
benefit of looking at many data points.}

\textbf{Tail test -} A tail test is one which depends only on forecasts
made eventually, after the calibration phase. Whereas much of the
literature emphasizes tests that provide their verdict at some finite
outcome, we take the opposite approach for some of our results and
consider comparison tests that are based on a long-run performance.
It turns out that the test proposed here, which is anonymous, error-free
and reasonable, is also a tail test. It is also unique in a very strong
sense---the error of any alternative tail test, which is also anonymous
and reasonable, can be made arbitrarily close to one.

\section{{\normalsize{}Model}}

{}{}At the beginning of each period $t=1,2,\ldots$ an outcome $\omega_{t}$,
drawn randomly by Nature from the set $\Omega=\{0,1\},$ is realized.\footnote{{}For expository reasons we restrict attention to a binary set $\Omega=\{0,1\}$.
The results extend to any finite set.} Before $\omega_{t}$ is realized, two self-proclaimed experts (sometimes
referred to as forecasters) simultaneously announce their forecast
in the form of a probability distribution over $\Omega$. We assume
that both forecasters observe all past outcomes and all previous pairs
of forecasts. For any (infinite) realization, $\omega\coloneqq\{\omega_{1},\omega_{2},\ldots\}\in\Omega^{\infty},$
we denote by $\omega^{t}\coloneqq\{\omega_{1},\omega_{2},\ldots,\omega_{t}\}$
its prefix of length $t$ (sometimes referred to as the partial history
of outcomes up to period $t$), and set $\omega^{0}\coloneqq\emptyset$.

We will abuse notation and use $\omega^{t}$ to denote the cylinder
set $\{\hat{\omega}\in\Omega^{\infty}|\ \hat{\omega}^{t}=\omega^{t}\}.$
In other words, $\omega^{t}$ will also denote the set of realizations
which share a common prefix of length $t$. For any $t$ we denote
by $g_{t}$ the $\sigma$-algebra on $\Omega^{\infty}$ generated
by the cylinder sets $\omega^{t}$ and let $g_{\infty}\coloneqq\sigma(\stackrel[t=0]{\infty}{\bigcup}g_{t})$
denote the smallest $\sigma$-algebra which consists of all cylinders
(also known as the Borel $\sigma$-algebra). Let $\Delta(\Omega^{\infty})$
be the set of all probability measures defined over the measurable
space $(\Omega^{\infty},g_{\infty})$.

At each stage, two forecasts (elements in $\Delta(\Omega)$) are provided
by two experts. Let $(\Omega\times\Delta(\Omega)\times\Delta(\Omega))^{t}$
be the set of all sequences composed of outcomes and pairs of forecasts
made up to time $t$ and let $\underset{t\geq0}{\bigcup}(\Omega\times\Delta(\Omega)\times\Delta(\Omega))^{t}$
be the set of all such finite sequences.

A (pure) forecasting strategy, $f,$ is a function that maps finite
histories to a probability distribution over $\Omega$. Formally,
$f\colon\underset{t\geq0}{\bigcup}(\Omega\times\Delta(\Omega)\times\Delta(\Omega))^{t}\longrightarrow\Delta(\Omega).$
Note that each forecast provided by one expert may depend, inter alia,
on those provided by the other expert in previous stages. Let $F$
denote the set of all forecasting strategies.

A probability measure $P\in\Delta(\Omega^{\infty})$ naturally induces
a (set of) corresponding forecasting strategy, denoted $f_{P}$, that
satisfies for any $\omega\in\Omega^{\infty}$ and $t>0$ such that
$P(\omega^{t})>0,$ 
\[
f_{P}(\omega^{t},\cdot,\cdot)[\omega_{t+1}]=P(\omega_{t+1}|\omega^{t}).
\]

\noindent Thus, the forecasting strategy $f_{P}$ derives its forecasts
from the original measure $P$ via Bayes rule. Note that this does
not restrict the forecast of $f_{P}$ over cylinders, $\omega^{t}$,
for which $P(\omega^{t})=0$.\footnote{Hereinafter we will often abuse notation and use $P$ instead of $f_{P}.$}

In the other direction, we abuse notation and given an ordered pair
of forecasting strategies, $f\coloneqq(f_{0},f_{1})\in F\times F$
(henceforth $f$), let $h$ be a function that maps each triplet $(\omega,f_{0},f_{1})$
to its uniquely induced play path: $h(\omega,f_{0},f_{1})\in(\varOmega\times\Delta(\Omega)\times\Delta(\Omega))^{\infty}$.
Additionally, for any $n\geq0$, the prefix (of length $n$) and the
suffix (starting at $n$) of $h(\omega,f_{0},f_{1})$ are denoted
by $h^{n}(\omega,f_{0},f_{1})$ and $h_{n}(\omega,f_{0},f_{1})$,
respectively. Whenever $(\omega,f_{0},f_{1})$ is clear from the context
we abuse notation and denote these by $h,\ h^{n},$ and $h_{n}$ respectively.

Now observe that a single forecasting strategy need not induce a measure
as its output may also depend on what another expert forecasts. However,
an ordered pair of forecasting strategies $f,$ does induce a pair
of probability measures, denoted $P_{0}^{f}(\cdot),P_{1}^{f}(\cdot)$,
over $\Omega^{\infty}$. By Kolomogorov's extension theorem, it is
enough to define these probabilities over the cylinder sets of the
form $\omega^{t},$ denoted for an arbitrary $\omega\in\Omega^{\infty},\ t>0$
and $i\in\{0,1\}\colon$

\begin{equation}
P_{i}^{f}(\omega^{t})=\stackrel[n=1]{t}{\prod}f_{i}(h^{n-1})[\omega_{n}].\label{eq:1 P(C_wt) =00003D00003D00003D00003D multiplef_i}
\end{equation}

\subsection{{\normalsize{}\label{subsec:first two axioms}Comparison test}}

{}A {\em comparison test} is a measurable function whose input
is a pair of two forecasting strategies and a realization, and whose
output is a (weak) order over the two experts. Formally, 
\[
T\colon\varOmega^{\infty}\times F\times F\longrightarrow\{0,\frac{1}{2},1\}
\]
where $T=i\neq\frac{1}{2}$ implies that expert $i$ is claimed as
better informed, while $T=\frac{1}{2}$ implies the test is inconclusive
(this cannot be avoided, for example, when both experts' forecasts
always agree).

A comparison test should, a priori, treat both experts similarly.
This is captured by the following notion of anonymity of a test.
\begin{defn}
A test $T$ is {\em anonymous} if for all $\omega\in\Omega^{\infty}$
and $f_{0},f_{1}\in F,$
\[
T(\omega,f_{0},f_{1})=1-T(\omega,f_{1},f_{0}).
\]
\end{defn}
In other words, the expert chosen by $T$ does not depend on the expert's
identity ($0$ or $1$). Note that whenever $f_{0}=f_{1},$ an anonymous
test $T$ must be inconclusive and always output $0.5$.

We follow the lion's share of the literature on single expert testing
and, furthermore, require that the outcome of the comparison test
depends only on predictions made along the realized play path. Formally, 
\begin{defn}
\label{Def  non-counterfactual definion}A test $T$ is {\em non-counterfactual}
if there exists a function 

\noindent 
\[
\hat{T}\colon(\varOmega\times\Delta(\Omega)\times\Delta(\Omega))^{\infty}\longrightarrow\{0,\frac{1}{2},1\}
\]
such that $T=\hat{T}\circ h.$ 
\end{defn}
\noindent Hereafter we restrict attention to anonymous and non-counterfactual
tests.

\subsection{{\normalsize{}\label{subsec:second two axioms}Desired properties}}

We now turn to formally define two desired properties for a comparison
test followed by the motivation. We will later argue that these induce
an essentially single comparison test.

The first property requires that whenever one of the experts has the
correct model (namely, he would have passed a standard single expert
test which has no type-1 errors) the test will surely not point at
the second expert as the superior one (albeit, it may provide a non-conclusive
outcome).

For any test, $T$, and an ordered pair of forecasting strategies,
$f,$ we denote by $\{T(\cdot,f)=k\}$ the set of realizations for
which the test outputs $k$.
\begin{defn}
\label{Def error-free}A test $T$ is {\em error-free} if for all
$f$ and $i\in\{0,1\},$ $P_{1-i}^{f}(\{T(\cdot,f)=i\})=0.$
\end{defn}
In other words, whenever one expert knows the probability distribution
governing the realizations of Nature, the test must not identify the
other expert  as the true expert. In the jargon of hypothesis testing,
Definition \ref{Def error-free} implies that an error-free test must
eliminate errors of type-1 and, symmetrically, type-2.

One trivial example of an error-free test is the test that constantly
outputs $\frac{1}{2}$. Note that it is also anonymous and non-counterfactual.
We shall later propose a non-trivial error-free test. Unfortunately
that test will also be indecisive at times but not always. In fact,
it turns out that error-free tests must be indecisive whenever the
experts induce a pair of measures that are mutually absolutely continuous.
Formally,
\begin{prop}
{}{}\label{prop:if Error-free then A.C} Let $f$ be such that $P_{1}^{f}\ll P_{0}^{f}.$
If \textup{$T$} is error-free then $P_{0}^{f}(\{T(\cdot,f)=0\})<1.$ 
\end{prop}
Thus, expert $0$, from his own perspective, cannot be confident that
the test will identify him as better informed. Combine this with the
definition of an error-free test to conclude that from the expert's
point of view that test must, at times, be inconclusive:
\begin{cor}
\label{Corollary a.c implies inconclusiveness}Let $f$ be such that
$P_{1}^{f}\ll P_{0}^{f}$. If \textup{$T$} is error-free then $P_{0}^{f}(\{T(\cdot,f)=\frac{1}{2}\})>0.$
\end{cor}
We now turn to the proof of Proposition \ref{prop:if Error-free then A.C}. 
\begin{proof}
Assume that 
\begin{equation}
P_{0}^{f}(\{T(\cdot,f)=0\})=1.\label{eq:4}
\end{equation}

\noindent Since $P_{1}^{f}\ll P_{0}^{f}$ it follows from $\eqref{eq:4}$
that 
\[
P_{0}^{f}(\{T(\cdot,f)=0\}^{c})=0\Longrightarrow P_{1}^{f}(\{T(\cdot,f)=0\}^{c})=0.
\]
Therefore 
\[
P_{1}^{f}(\{T(\cdot,f)=0\})=1,
\]
which by the anonymity of $T$ contradicts the assumption that $T$
is error-free.\footnote{In the context of hypothesis testing, Corollary \ref{Corollary a.c implies inconclusiveness}
implies that an error-free test will not have a power of one whenever
the null hypothesis is absolutely continues w.r.t to the alternative
one.}
\end{proof}
The next property of a comparison test asserts that for any set of
realizations assigned zero probability by one forecaster and positive
probability by the other forecaster, there must be some subset of
realizations for which that other forecaster is deemed superior. Formally, 
\begin{defn}
\label{Def resonable}A test $T$ is {\em reasonable} if for all
$f$ and $i\in\{0,1\},$ and for all measurable set $A$,

\begin{equation}
P_{i}^{f}(A)>0\text{ and }P_{1-i}^{f}(A)=0\implies P_{i}^{f}(A\cap\{T(\cdot,f)=i\})>0.\label{eq:reasonableness condition}
\end{equation}
\end{defn}
\noindent It should be emphasized that reasonableness and error-free
are not related notions. To see why error-free does not imply reasonableness,
just consider the constant error-free test $T\equiv\frac{1}{2}$.
An example of a reasonable test that is not error-free is deferred
to the end of Subsection \ref{subsec:Tail-Test}.

\section{{\normalsize{}The derivative test}}

We now turn to our construction of a non-counterfactual, anonymous,
error-free and reasonable comparison test. Before doing so, some preliminaries
are required.

Given an ordered pair of forecasting strategies, $f,$ a realization
of Nature, $\omega\in\Omega^{\infty}$, we define the likelihood ratio
between the two forecasters at time $t$ as, 
\[
D_{f_{0}}^{t}f_{1}(\omega)\coloneqq\stackrel[n=1]{t}{\prod}\frac{f_{1}(h^{n-1})[\omega_{n}]}{f_{0}(h^{n-1})[\omega_{n}]}.
\]
 Define the following limit functions:

\[
\overline{D}_{f_{0}}f_{1}(\omega)\coloneqq\begin{cases}
\begin{array}{l}
\underset{t\rightarrow\infty}{limsup}D_{f_{0}}^{t}f_{1}(\omega),\\
+\infty,
\end{array} & \begin{array}{l}
f_{0}(h^{n-1})[\omega_{n}]>0\text{ for all }n\geq1\\
f_{0}(h^{n-1})[\omega_{n}]=0\text{ for some }n.
\end{array}\end{cases}
\]

\[
\underline{D}{}_{f_{0}}f_{1}(\omega)\coloneqq\begin{cases}
\begin{array}{l}
\underset{t\rightarrow\infty}{liminf}D_{f_{0}}^{t}f_{1}(\omega),\\
+\infty,
\end{array} & \begin{array}{l}
f_{0}(h^{n-1})[\omega_{n}]>0\text{ for all }n\geq1\\
f_{0}(h^{n-1})[\omega_{n}]=0\text{ for some }n.
\end{array}\end{cases}
\]

Whenever the two limits coincide and take a finite value, we refer
to this value as the derivative of the forecasting strategy $f_{1}$
with respect to the forecasting strategy $f_{0}$ at $\omega$. Formally,
if $\overline{D}_{f_{0}}f_{1}(\omega)=\underline{D}{}_{f_{0}}f_{1}(\omega)<\infty,$
let $D_{f_{0}}f_{1}(\omega)=\overline{D}_{f_{0}}f_{1}(\omega)$ be
the {\em derivative} of $f_{1}$ with respect to $f_{0}$ at $\omega$.
We are now ready to define the {\em derivative test}, denoted ${\cal D},$
a non-counterfactual and anonymous test which we will show is error-free
and reasonable:

\begin{equation}
{\cal D}(\omega,f_{0},f_{1})=\begin{cases}
\begin{array}{l}
1,\\
0.5,\\
0,
\end{array} & \begin{array}{l}
D_{f_{1}}f_{0}(\omega)=0\\
\text{other}\\
D_{f_{0}}f_{1}(\omega)=0.
\end{array}\end{cases}\label{eq:the test}
\end{equation}

Expert $i$ is indicated as the true forecaster at $\omega$ whenever
the derivative of $f_{1-i}$ with respect to $f_{i}$ exists and equals
0. Intuitively, this happens when the probability assigned by expert
$i$ to the actual realization is infinitely larger than that assigned
by expert $1-i$.

It is obvious that ${\cal D}$  is non-counterfactual and anonymous.
We turn to prove that ${\cal D}$ is also error-free and reasonable.
To do so, we will need the following two technical observations regarding
derivatives of forecasting strategies:
\begin{lem}
{}{}\label{Lemma:liminf=00003Dlimsup} Fix $0<\alpha<\infty$ and
let $A\subset\Omega^{\infty}$ be a measurable set. Then

{}{}$a)\ \ A\subset\{\omega|\ \underline{D}{}_{f_{0}}f_{1}(\omega)\leq\alpha\}\Longrightarrow P_{1}^{f}(A)\leq\alpha P_{0}^{f}(A).$

{}{}$b)\ \ A\subset\{\omega|\ \overline{D}_{f_{0}}f_{1}(\omega)\geq\alpha\}\Longrightarrow P_{1}^{f}(A)\geq\alpha P_{0}^{f}(A).$ 
\end{lem}
\vspace{0cm}

\begin{lem}
\label{Lemma: D exists and finite}For all $f,$ $D_{f_{0}}f_{1}$
exists and is finite $P_{0}^{f}$ - a.e. 
\end{lem}
The proof of Lemmas \ref{Lemma:liminf=00003Dlimsup} and \ref{Lemma: D exists and finite}
are relegated to  Appendix \ref{Appendix A sec:Missing-Proofs}.

\subsection{{\normalsize{}{}{}{}{}{}${\cal D}$ is error-free and reasonable}}

Now that we have established the existence and the finiteness of the
test ${\cal D},$ let us prove it complies with the two central properties
for comparison tests: 
\begin{thm}
\label{Theorem 1 D is error free and reasonable}The derivative test,
${\cal D}$, is a reasonable and error-free test. 
\end{thm}
\begin{proof}
\noindent \textbf{Part 1 - ${\cal D}$ is reasonable:} Let $A$ be
a measurable set and assume (w.l.o.g) that 
\begin{equation}
P_{0}^{f}(A)>0\text{ and }P_{1}^{f}(A)=0.\label{eq:15}
\end{equation}
For $a>0$ let us denote $R_{a}\coloneqq A\cap\{\omega|\ 0<a\leq D_{f_{0}}f_{1}(\omega)<\infty\}.$
Note that if $P_{0}^{f}(R_{a})>0$ then applying part $b$ of Lemma
\ref{Lemma:liminf=00003Dlimsup} yields

\noindent 
\[
P_{1}^{f}(R_{a})\geq aP_{0}^{f}(R_{a})>0
\]

\noindent which contradicts $\eqref{eq:15}.$ Therefore,

%\noindent Therefore, 

\[
P_{0}^{f}(A\cap\{\omega|\ 0<D_{f_{0}}f_{1}(\omega)<\infty\})=P_{0}^{f}(\underset{\underset{a\in\mathbb{Q}}{0<a}}{\bigcup}R_{a})\leq\underset{\underset{a\in\mathbb{Q}}{0<a}}{\sum}P_{0}^{f}(R_{a})=0.
\]

\noindent Since, by Lemma \ref{Lemma: D exists and finite}, $D_{f_{0}}f_{1}$
exists and is finite $P_{0}^{f}-a.e.$, we conclude that

\[
P_{0}^{f}(A\cap\{\omega|\ D_{f_{0}}f_{1}(\omega)=0\}^{c})=0.
\]

\noindent Hence,

\begin{equation}
0<P_{0}^{f}(A)=P_{0}^{f}(A\cap\{\omega|\ D_{f_{0}}f_{1}(\omega)=0\})=P_{0}^{f}(A\cap\{{\cal D}(\cdot,f)=0\}),\label{eq:P(A intersect A_T0)=00003DP(A)}
\end{equation}

\noindent where the right-most equality follows from \eqref{eq:the test}.
Inequality \eqref{eq:P(A intersect A_T0)=00003DP(A)} implies that
the test ${\cal D}$ is reasonable.

\textbf{Part 2 - ${\cal D}$ is error-free:} Note (w.l.o.g) that

\[
\begin{array}{l}
\{{\cal D}(\cdot,f)=1\}\\
\\
=\{\omega|\ \underset{t\rightarrow\infty}{lim}\,D_{f_{1}}^{t}f_{0}(\omega)=0\text{ and }f_{1}(h^{n-1})[\omega_{n}]>0\text{ for all }n\geq1\}\\
\\
\subset\{\omega|\ \underset{t\rightarrow\infty}{lim}\,D_{f_{0}}^{t}f_{1}(\omega)=\infty\}\cup\{\omega|\ f_{0}(h^{n-1})[\omega_{n}]=0\text{ for some \ensuremath{n}}\}\\
\\
\subset\{\omega|\ \underline{D}{}_{f_{0}}f_{1}(\omega)=\overline{D}_{f_{0}}f_{1}(\omega)=\infty\}.
\end{array}
\]

\noindent By Lemma \ref{Lemma: D exists and finite}, $D_{f_{0}}f_{1}$
is finite $P_{0}^{f}-a.e.;$ thus

\[
P_{0}^{f}(\{{\cal D}(\cdot,f)=1\})\leq P_{0}^{f}(\{\omega|\ \underline{D}{}_{f_{0}}f_{1}(\omega)=\overline{D}_{f_{0}}f_{1}(\omega)=\infty\})=0,
\]

\noindent and ${\cal D}$ is error-free.
\end{proof}
\vspace{0cm}

\begin{rem}
The test ${\cal D}$ and its key properties can be usefully viewed
as an implication of the Lebesgue decomposition (\Citet*{Billingsley},
Section 31). A standard decomposition usually involves a decomposition
of one measure with respect to another into a singular part and an
absolutely continuous part. Here, it is applied in both directions
in such a way that allows some flexibility on how measure-zero sets
are handled. Given a pair of forecasting strategies $f,$ we decompose
the set $\Omega^{\infty}$ into three sets: $\{{\cal D}(\cdot,f)=1\}$
which corresponds to expert $1$'s induced measure $P_{1}^{f},$ $\{{\cal D}(\cdot,f)=0\}$
which corresponds to expert $0$'s induced measure $P_{0}^{f},$ and
$\{{\cal D}(\cdot,f)=\frac{1}{2}\}$ where the measures are mutually
absolutely continuous. The outcome of the test is found accordingly.
\end{rem}

\subsection{{\normalsize{}The uniqueness of ${\cal D}$}}

Although there may be other error-free and reasonable comparison tests
they are essentially equivalent to the derivative test. To capture
this idea we introduce the following equivalence relation over tests:
\begin{defn}
\label{Def: different test }We say that {\em the test $T$ is equivalent
to the test $\hat{T}$ with respect to the pair of forecasters $f,$}
denoted $T\sim_{f}\hat{T}$, if and only if for all $i\in\{0,1\},$
\[
P_{i}^{f}(\{\omega|\;T(\omega,f_{0},f_{1})\neq\hat{T}(\omega,f_{0},f_{1})\})=0.
\]
{\em $T$ is equivalent to the test $\hat{T}$}, denoted {\em
$T\sim\hat{T}$}, if and only if $T$ is equivalent to the test $\hat{T}$
with respect to any pair of forecasters.
\end{defn}
\vspace{0cm}

\begin{prop}
\label{prop:equivalence}The relation $\sim$ is an equivalence relation
over the set of all comparison tests.
\end{prop}
The proof of Proposition \ref{prop:equivalence} is relegated to Appendix
\ref{Appendix A sec:Missing-Proofs}. To establish the theorem about
the essential uniqueness of the derivative test we will consider an
arbitrary anonymous, non-counterfactual, reasonable test, $T$, that
is not equivalent to ${\cal D}$. We will then argue that $T$ cannot
be error-free. We will do so by constructing a pair of forecasting
strategies for which the error-free condition fails.\footnote{In fact we show a much stronger result; Theorem \ref{thm:main Reasonable implies not error-free}
asserts that $T$ admits an error with respect to \textbf{any} pair
of forecasting strategies for which $T$ is not equivalent to ${\cal D}$.}
\begin{thm}
\label{thm:main Reasonable implies not error-free}Let $T$ be an
anonymous, non-counterfactual, reasonable test. If $T\nsim{\cal D}$
then $T$ is not error-free.
\end{thm}
\begin{proof}
Assume by contradiction that $T$ is error-free. Let $f$ be such
that $T\nsim_{f}{\cal D}$, then (w.l.o.g.)  $\exists k,l(\neq k)\in\{0,\frac{1}{2},1\}$
such that

\[
P_{0}^{f}(\{T(\cdot,f)=l\}\cap\{{\cal D}(\cdot,f)=k\})>0.
\]

\noindent In addition, by Part 2 of Theorem \ref{Theorem 1 D is error free and reasonable},
${\cal D}$ is error-free; therefore 

\[
P_{0}^{f}(\{T(\cdot,f)=1\})=P_{0}^{f}(\{{\cal D}(\cdot,f)=1\})=0
\]

\noindent and consequently,

\[
P_{0}^{f}(A_{1}\coloneqq\{T(\cdot,f)=0\}\cap\{{\cal D}(\cdot,f)=\frac{1}{2}\})>0\ \text{ or }\;P_{0}^{f}(A_{2}\coloneqq\{T(\cdot,f)=\frac{1}{2}\}\cap\{{\cal D}(\cdot,f)=0\})>0.
\]

Case 1$\colon$ $P_{0}^{f}(A_{1})>0$. By Part 1 of Theorem \ref{Theorem 1 D is error free and reasonable},
${\cal D}$ is reasonable; thus

\[
P_{1}^{f}(A_{1})=0\Longrightarrow P_{0}^{f}(A_{1}\cap\{{\cal D}(\cdot,f)=0\})>0
\]

\noindent which leads to a contradiction, since $\{{\cal D}(\cdot,f)=0\},\{{\cal D}(\cdot,f)=\frac{1}{2}\}$
are disjoint. Thus

\noindent 
\[
P_{1}^{f}(\{T(\cdot,f)=0\})>0
\]
which contradicts the assumption that $T$ is error-free.

Case 2$\colon$ $P_{0}^{f}(A_{2})>0.$ By the assumption, $T$ is
a reasonable test where, by Part 2 of Theorem \ref{Lemma:liminf=00003Dlimsup},
${\cal D}$ is error-free; therefore the contradiction

\[
P_{1}^{f}(\{{\cal D}(\cdot,f)=0\})>0
\]

\noindent follows analogously from Case 1.
\end{proof}

\section{{\normalsize{}\label{subsec:Tail-Test} Tail tests}}

In the introduction, we state our intention to study tests in which
decisions are made for the distant future. In this section we take
this a step further and consider tests which not only enable decisions
to be made for the distant future, but also \textbf{only} for the
distant future.

The motivation for this is that a tester must allow the two forecasters
(some time) to accumulate data so they can calibrate their model.
A forecaster may have a very good parametric model in mind but can
only calibrate the values of the parameters by observing enough data.
A test that allows for such an initial calibration stage is called
a {\em tail test}. Formally,
\begin{defn}
The pair of triplets, $(\omega,f_{0},f_{1}),(\tilde{\omega},\tilde{f_{0}},\tilde{f})\in\Omega^{\infty}\times F\times F,$
{\em eventually coincide} if there exists $n>1$ such that for all
$1\leq t\leq n-1,\ i\in\{0,1\},$

\begin{equation}
h_{n}=\tilde{h}_{n}\text{ and }f_{i}(h^{t-1})[\omega_{t}]>0,\tilde{f_{i}}(\tilde{h}^{t-1})[\tilde{\omega}_{t}]>0\label{eq:16-1 tail condition}
\end{equation}

\noindent (where $\tilde{h}_{n}\coloneqq h_{n}(\tilde{\omega},\tilde{f_{0}},\tilde{f_{1}}),\ \tilde{h}^{t-1}\coloneqq h^{t-1}(\tilde{\omega},\tilde{f_{0}},\tilde{f_{1}})$).
\end{defn}
\vspace{0cm}

In words, the two play paths agree from some time on whenever the
prefix has mutually positive probability.
\begin{defn}
\label{Def tail test}$T$ is a {\em tail test} if whenever a pair
of triplets, $(\omega,f_{0},f_{1}),(\tilde{\omega},\tilde{f_{0}},\tilde{f})\in\Omega^{\infty}\times F\times F$,
eventually coincide, then $T(\omega,f_{0},f_{1})=T(\tilde{\omega},\tilde{f_{0}},\tilde{f_{1}}).$

In layman's terms, a tail test ignores the prefix of the sequence
and makes the comparison between the two experts based on the suffix
of forecasts and realizations.
\end{defn}
It turns out that the derivative test also conforms with the tail
property:
\begin{thm}
${\cal D}$ is a tail test.
\end{thm}
\begin{proof}
Let $(\omega,f_{0},f_{1}),(\omega',f_{0}',f_{1}')\in\Omega^{\infty}\times F\times F$
be a pair of triplets that eventually coincide for some $n>1$. Let
$(\omega'',f_{0}'',f_{1}'')\in\Omega^{\infty}\times F\times F$ be
a triplet that satisfies

\begin{equation}
h_{1}(\omega'',f_{0}'',f_{1}'')=h_{n}(\omega,f_{0},f_{1}).\label{eq: in the proof T_D is a tail test}
\end{equation}

\noindent Since, by the right part of $\eqref{eq:16-1 tail condition},$
$D_{f_{0}}^{t}f_{1}(\omega)>0$ for all $1\leq t<n-1,$ it follows
from \eqref{eq: in the proof T_D is a tail test} that\footnote{Note that $D_{f_{0}''}f_{1}''(\omega'')=0$ if and only if $D_{f_{0}}f_{1}(\omega'')$
exists and equals $0.$ }

\[
0={\cal D}(\omega'',f_{0}'',f_{1}'')\iff0=D_{f_{0}}^{n-1}f_{1}(\omega)\cdot D_{f_{0}''}f_{1}''(\omega'')=D_{f_{0}}f_{1}(\omega)\iff{\cal D}(\omega,f_{0},f_{1})=0.
\]

\noindent Additionally, by the same consideration we have

\[
1={\cal D}(\omega'',f_{0}'',f_{1}'')\iff0=D_{f_{1}''}f_{0}''(\omega'')\iff{\cal D}(\omega,f_{0},f_{1})=1,
\]

\noindent and therefore

\[
{\cal D}(\omega'',f_{0}'',f_{1}'')={\cal D}(\omega,f_{0},f_{1}).
\]

\noindent Similarly, we show that ${\cal D}(\omega'',f_{0}'',f_{1}'')={\cal D}(\omega',f_{0}',f_{1}')$
by replacing $h_{n}(\omega,f_{0},f_{1})$ with $h_{n}(\omega',f_{0}',f_{1}')$
in \eqref{eq: in the proof T_D is a tail test} and this concludes
the proof. 
\end{proof}
To establish that ${\cal D}$ is unique among all reasonable error-free
tests, we showed that for an arbitrary non-equivalent yet reasonable
test there must be some error. What we have shown is that there is
a pair of experts where one expert will assign a positive probability
to the test pointing at the other expert as more informative. That
probability, the error probability, although positive is possibly
very small. It turns out that if we restrict the discussion to tail
tests, the uniqueness of ${\cal D}$ comes in a stronger form, as
the error probability can be made arbitrarily close to one. In other
words, for an arbitrary non-equivalent reasonable \textbf{tail} test
and $0<\epsilon<1$, there exists a pair of experts for which one
expert assigns a probability of $1-\epsilon$ to the other expert
being deemed more informative. 

Before we state this theorem we require the following lemma. 
\begin{lem}
\label{Lemma  if T reasonable then P1(A intersect AT0)}If $T$ is
reasonable then for all $f$ and $i\in\{0,1\},k\neq i,$ and for all
measurable set $A,$
\end{lem}
\[
P_{i}^{f}(A\cap\{T(\cdot,f)=k\})>0\implies P_{1-i}^{f}(A\cap\{T(\cdot,f)=k\})>0.
\]

\begin{proof}
Let $A$ be a measurable set and (w.l.o.g) assume by contradiction
that

\[
P_{1}^{f}(A\cap\{T(\cdot,f)=k\})>0\text{ and }P_{0}^{f}(A\cap\{T(\cdot,f)=k\})=0
\]
for some $k\in\{0,\frac{1}{2}\}.$ $T$ is reasonable; thus \eqref{eq:reasonableness condition}
yields $P_{1}^{f}(A\cap\{T(\cdot,f)=k\}\cap\{T(\cdot,f)=1\})>0$ which
contradicts the fact that $\{T(\cdot,f)=k\},\{T(\cdot,f)=1\}$ are
disjoint sets. 
\end{proof}
Now we turn to establish a strong version of the uniqueness of ${\cal D\colon}$
\begin{thm}
\label{Theorem, main 2}Let $T$ be an anonymous, non-counterfactual,
reasonable tail test. If $T\nsim{\cal D}$ then for all $0<\epsilon<1$
there exists $\hat{f}\coloneqq(\hat{f_{0}},\hat{f_{1}})$ such that

\[
P_{0}^{\hat{f}}(\{T(\cdot,\hat{f})=1\})>1-\epsilon\text{ or }P_{1}^{\hat{f}}(\{T(\cdot,\hat{f})=0\})>1-\epsilon.
\]
\end{thm}
\begin{proof}
By Theorem \ref{thm:main Reasonable implies not error-free} (w.l.o.g.)
there exists a pair $f\coloneqq(f_{0},f_{1})$ such that $P_{1}^{f}(\{T(\cdot,f)=0\})>0$.
In addition, since $\{T(\cdot,f)=0\}$ is $g_{\infty}-measurable$
we can apply the Levy upwards theorem (\citet{Probability-with-Martingales},
Theorem 14.2.) to obtain

\[
\begin{array}{l}
\underset{t\rightarrow\infty}{lim}P_{1}^{f}(\{T(\cdot,f)=0\}|\:g_{t})\\
\\
=\underset{t\rightarrow\infty}{lim}E^{P_{1}^{f}}[{\ensuremath{\mathbf{1}}}_{\{T(\cdot,f)=0\}}|\:g_{t}]=E^{P_{1}^{f}}[{\ensuremath{\mathbf{1}}}_{\{T(\cdot,f)=0\}}|\:g_{\infty}]={\ensuremath{\mathbf{1}}}_{\{T(\cdot,f)=0\}},\ P_{1}^{f}-a.s.
\end{array}
\]

Therefore, there exists $B^{f}\subset\{T(\cdot,f)=0\}$ with $P_{1}^{f}(B^{f})=P_{1}^{f}(\{T(\cdot,f)=0\})$
such that for all $\omega\in B^{f}$ and for all $t\ge1,$

\begin{equation}
\underset{t\rightarrow\infty}{lim}P_{1}^{f}(\{T(\cdot,f)=0\}|\:\omega^{t})=1\text{ and }f_{1}(h^{t-1}(\omega,f_{0},f_{1}))[\omega_{t}]>0.\label{eq:in theorem 2}
\end{equation}

Let $0<\epsilon<1.$ Fix $\tilde{\omega}\in B^{f}$ and observe that
from \eqref{eq:in theorem 2} there exists $n=n_{(\epsilon,\tilde{\omega},f)}>1$
such that for all $t\geq n-1,$

\[
P_{1}^{f}(\{T(\cdot,f)=0\}\cap\tilde{\omega}^{t})>(1-\epsilon)P_{1}^{f}(\tilde{\omega}^{t})>0.
\]

\noindent Thus, applying Lemma \ref{Lemma  if T reasonable then P1(A intersect AT0)}
 yields $P_{0}^{f}(\{T(\cdot,f)=0\}\cap\tilde{\omega}^{n-1})>0$ and
consequently,

\noindent $f_{0}(h^{t-1}(\tilde{\omega},f_{0},f_{1}))[\tilde{\omega}_{t}]>0$
for all $1\leq t\leq n-1,$ is inferred from \eqref{eq:1 P(C_wt) =00003D00003D00003D00003D multiplef_i}.
Now, modify $f$ to be the forecasting strategy $\hat{f}$ which one-step-ahead
conditionals satisfy\footnote{Note that the corresponding forecasting strategy $\hat{f_{i}}$ determines
the one-step-ahead forecasts up to time $n$ only through the history
of outcomes and does not depend on the full histories.}

\[
\hat{f_{i}}(\omega^{t-1},\cdot,\cdot)[\omega_{t}]=\begin{cases}
\begin{array}{l}
1,\\
f_{i}(\omega^{t-1},\cdot,\cdot)[\omega_{t}],\\
0,
\end{array} & \begin{array}{l}
\omega^{t}=\tilde{\omega}^{t},\ t<n\\
\text{other}\\
\omega^{t}\neq\tilde{\omega}^{t},\ t<n.
\end{array}\end{cases}
\]

\noindent Observe that, by construction, for all $\omega\in\{T(\cdot,f)=0\}\cap\tilde{\omega}^{n-1}$
we obtain $h_{n}(\omega,f_{0},f_{1})=h_{n}(\omega,\hat{f_{0}},\hat{f_{1}}),$
and in addition to that, for all $1\leq t\leq n-1,\ i\in\{0,1\},$

\[
f_{i}(h^{t-1}(\omega,f_{0},f_{1}))[\omega_{t}]=f_{i}(h^{t-1}(\tilde{\omega},f_{0},f_{1}))[\omega_{t}]>0,\ \hat{f_{i}}(h^{t-1}(\omega,\hat{f_{0}},\hat{f_{1}}))[\omega_{t}]=\hat{f_{i}}(h^{t-1}(\tilde{\omega},\hat{f_{0}},\hat{f_{1}}))[\tilde{\omega}_{t}]>0.
\]

\noindent Hence, $(\omega,f_{0},f_{1}),(\omega,\hat{f_{0}},\hat{f_{1}})$
eventually coincide by \eqref{eq:16-1 tail condition} and since $T$
is a tail test it follows that $T(\omega,\hat{f_{0}},\hat{f_{1}})=T(\omega,f_{0},f_{1})=0$
yielding $\omega\in\{T(\cdot,\hat{f})=0\}\cap\tilde{\omega}^{n-1}.$ As a result,

\[
\begin{array}{l}
P_{1}^{\hat{f}}(\{T(\cdot,\hat{f})=0\})\\
\\
=P_{1}^{\hat{f}}(\{T(\cdot,\hat{f})=0\}|\:\tilde{\omega}^{n-1})\geq P_{1}^{\hat{f}}(\{T(\cdot,f)=0\}|\:\tilde{\omega}^{n-1})=P_{1}^{f}(\{T(\cdot,f)=0\}|\:\tilde{\omega}^{n-1})>1-\epsilon,
\end{array}
\]

\noindent and therefore completes the proof.
\end{proof}
Unfortunately, as the next example shows, this strong version of uniqueness
cannot be established without resorting to tail tests. The same example
also serves to demonstrate that a reasonable test is not necessarily
error-free.
\begin{example}
Assuming that from day two onward, along a realization $\overset{1}{\omega}\coloneqq(1,1,1,,,),$
two forecasting strategies are shown to have similar predictions,
according to an IID distribution with parameter $1,$ where on day
one, one expert assigns 1 to the outcome 1 whereas the other expert
assigns half. Let $\overrightarrow{h},\overleftarrow{h}$ denote the
corresponding uniquely induced play paths and consider the following
test:

\[
T(\omega,f_{0},f_{1})=\begin{cases}
\begin{array}{l}
{\cal D}(\omega,f_{0},f_{1}),\\
0,\\
1,
\end{array} & \begin{array}{l}
\text{other}\\
h=\overleftarrow{h}\\
h=\overrightarrow{h}.
\end{array}\end{cases}
\]

Note, for every triplet $(\omega,f_{0},f_{1})$, whose induced play
path coincides with $\overrightarrow{h}$ or $\overleftarrow{h}$,
there exists $i\in\{0,1\}$ such that

\begin{equation}
P_{i}^{f}(\{T(\cdot,f)=1-i\})=P_{i}^{f}(\{\overset{1}{\omega}\})=\frac{1}{2}<1\label{eq: reasonable and not error free, in exaple-1-1}
\end{equation}
where the most-left equality holds, since ${\cal D}$ is error-free.
Moreover, since $P_{i}^{f}(\{\overset{1}{\omega}\})>0$ for all $i\in\{0,1\}$
and ${\cal D}$ is a reasonable test, it follows that $T$ is reasonable
even as it admits a bounded error by \eqref{eq: reasonable and not error free, in exaple-1-1}.
The fact that $T$ is not a tail test follows directly from the anonymity
of $T$ along $\overrightarrow{h},\overleftarrow{h}.$ 
\end{example}

\section{{\normalsize{}\label{sec:Ideal-tests}Ideal tests}}

Recall that an error-free test eliminates the occurrences in which
the less-informed expert is pointed out. A stronger and more appealing
property is to point out the better-informed expert. Informally, we
would like to consider tests that have the following property: $P_{i}^{f}(\{T(\cdot,f)=i\})=1$
whenever $f_{0}\not=f_{1}$. However, there could be pairs of forecasters
that are not equal but induce the same probability distribution. 
\begin{defn}
A test $T$ is {\em ideal with respect to $W\subseteq F$} if for
all $f\in W\times W$ and $i\in\{0,1\}$ such that $P_{i}^{f}\neq P_{1-i}^{f},$

\[
P_{i}^{f}(\{T(\cdot,f)=i\})=1.
\]

\noindent It is called {\em ideal} if it is {\em ideal} with
respect to $F.$
\end{defn}
In other words, whenever expert $i$ knows the actual data generating
process and expert $1-i$ does not, an ideal test will surely identify
the informed expert. In addition, it is a straightforward corollary
of Proposition \ref{prop:if Error-free then A.C} that there exists
no ideal test with respect to a set of forecasting strategies whenever
one induced measure is absolutely continuous with respect to the other. 

It is a common notion that two measures $P,Q$ are mutually singular
with respect to each other, denoted $P\perp Q,$ if there exists a
set $A$ such that $P(A)=Q(A^{c})=1.$ 
\begin{defn}
\label{Def Two-forecasting-strategies, M.S}Two forecasting strategies,
$f_{0},f_{1}\in F,$ are said to be mutually singular with respect
to each other, if $P_{0}^{f}\perp P_{1}^{f}.$ A set $W\subseteq F$
is pairwise mutually singular if for any pair $f\in W\times W$ such
that $P_{0}^{f}\neq P_{1}^{f}\colon f_{0},f_{1}$ are mutually singular
with respect to each other. 
\end{defn}
In other words, two forecasting strategies are mutually singular with
respect to each other if their corresponding induced measures are
mutually singular with respect to each other. The next lemma asserts
that a reasonable test is able to perfectly distinguish between `far'
measures which are induced from forecasting strategies which are said
to be mutually singular with respect to each other. 
\begin{lem}
{}{}\label{lem:reasonable implies m.s}Let $f_{0},f_{1}\in F$ be
mutually singular with respect to each other. If $T$ is reasonable
then for all $i\in\{0,1\},$

\[
P_{i}^{f}(\{T(\cdot,f)=i\})=1.
\]
\end{lem}
The proof of Lemma \ref{lem:reasonable implies m.s} is relegated
to Appendix \ref{Appendix A sec:Missing-Proofs}. It should be noted
that Lemma \ref{lem:reasonable implies m.s} holds even for $T$ which
is not error-free. 

\vspace{0cm}

The next theorem provides a necessary and sufficient condition for
the existence of an ideal test over sets. 
\begin{thm}
\label{Theorem m.s caraterization}There exists a non-counterfactual
anonymous ideal test with respect to $W$ if and only if \textup{$W$
is} pairwise mutually singular.
\end{thm}
\begin{proof}
$\Longleftarrow$From Lemma \ref{lem:reasonable implies m.s} and
Part 1 of Theorem \ref{Theorem 1 D is error free and reasonable}
we conclude that ${\cal D}$ is an ideal test with respect to $W$.

$\Longrightarrow$Let $T$ be a non-counterfactual anonymous ideal
test with respect to a set $W.$ Let $f\in W\times W$ be such that
$P_{0}^{f}\neq P_{1}^{f}$ and observe that since $\{T(\cdot,f)=0\},\{T(\cdot,f)=1\}$
are disjoint and $T$ is ideal, we obtain

\[
1=P_{i}^{f}(\{T(\cdot,f)=i\})=P_{1-i}^{f}(\{T(\cdot,f)=i\}^{c})
\]

\noindent for all $i\in\{0,1\},$ yielding $W$ that is pairwise mutually
singular.
\end{proof}
We conclude the section with an example of a test over a domain of
mutually singular forecasts:
\begin{example}
Let

\[
W_{IID}\times W_{IID}\coloneqq\{f|\ \forall i\in\{0,1\}\ \exists a_{f_{i}}\in[0,1]\ s.t\ \forall\omega\in\Omega^{\infty},\ f_{i}(\omega^{t},\cdot,\cdot)[1]\equiv a_{f_{i}}\}.
\]

\noindent For $\omega\in\Omega^{\infty}$ denote the average realization
by

\noindent 
\[
a_{\omega}\coloneqq\underset{t\rightarrow\infty}{lim}\left(\frac{\stackrel[n=1]{t}{\sum}1_{\{\omega_{n}=1\}}}{t}\right)
\]
(whenever the limit exists) and consider the following comparable
test 
\[
T(\omega,f_{0},f_{1})=\begin{cases}
\begin{array}{l}
1,\\
0.5,\\
0,
\end{array} & \begin{array}{l}
f_{1}(h^{0})[1]=a_{\omega}\neq f_{0}(h^{0})[1]\\
\text{other}\\
f_{0}(h^{0})[1]=a_{\omega}\neq f_{1}(h^{0})[1].
\end{array}\end{cases}
\]
\end{example}
\noindent Obviously, $T$ is well-defined, anonymous and non-counterfactual.
Showing that $T$ is ideal with respect to $W_{IID}$ is a mere application
of the law of large numbers.

\section{{\normalsize{}\label{sec:Existing-Tests}Existing tests}}

It is natural to inquire whether comparison tests previously proposed
comply with the properties we introduced. We turn to discuss the tests
proposed in \Citet*{Al-Najjar-2008} and \Citet*{Feinberg-Stewart-2008}.
It turns out that neither of these tests satisfies the full axiomatic
system which was introduced in Subsections \ref{subsec:first two axioms}
and \ref{subsec:second two axioms}, and hence does not belong to
the equivalence class represented by ${\cal D}.$ 

\subsection{{\normalsize{}\label{subsec:The-likelihood-ratio}The likelihood
ratio test}}

Al-Najjar and Weinstein (2008) introduced the following test:

\[
L(\omega,f_{0},f_{1})=\begin{cases}
\begin{array}{l}
1,\\
0.5,\\
0,
\end{array} & \begin{array}{l}
\underset{t\rightarrow\infty}{liminf}D_{f_{0}}^{t}f_{1}(\omega)>1\\
\text{other}\\
\underset{t\rightarrow\infty}{limsup}D_{f_{0}}^{t}f_{1}(\omega)<1.
\end{array}\end{cases}
\]

In other words, a likelihood ratio of one suggests that both experts
are likely equal  and so the test cannot determine which is better.
Similarly, the same conclusion holds whenever the likelihood ratio
oscillates infinitely often below and above one. Otherwise, if the
likelihood ratio is eventually greater than one (smaller than one)
then expert $1$ (expert $0$) is deemed superior. Note that this
test differs from ${\cal D}$ whenever the likelihood ratio is high
but finite. In our case, the test does not prefer any expert, whereas
the test $L$ does. It turns out that this test does not satisfy all
the properties we introduce:
\begin{claim}
\label{claim 1: al-najar}L is reasonable and is not error-free.
\end{claim}
\begin{proof}
Let $f_{1}$ be a forecasting strategy which deterministically predicts
$\overset{1}{\omega.}$ Let $0<\epsilon<1$ and let $f_{0}$ be the
forecasting strategy which predicts $(1-\epsilon)$ at day one and
meets $f_{1}$ from day two onward regardless of any past history.
Note that if $P_{0}^{f}$ is the true measure, then $L(\stackrel{1}{\omega}f_{0},f_{1})=\frac{1}{1-\epsilon}>1$
yielding $P_{0}^{f}(\{L(\cdot,f)=1\})\geq1-\epsilon.$ As a result,
since $\epsilon$ is taken arbitrarily, not only is $L$ not error-free
but it admits an arbitrarily large error. The fact that $L$ is reasonable
follows directly from Part 1 of Theorem \ref{Theorem 1 D is error free and reasonable}.
\end{proof}
\vspace{0cm}

\subsection{{\normalsize{}\label{subsec:The-Cross-Calibration-test}The cross-calibration
test}}

The cross-calibration  test introduced in Feinberg and Stewart (2008)
checks the empirical frequencies of the realization conditional on
each profile of forecasts that\textbf{ }occurs infinitely often (please
refer to Appendix \ref{sec:Appendix B The-Cross-Calibration} for
a formal definition). The test outputs a binary verdict (pass/fail)
for each of the experts separately, but does not rank them; nevertheless,
it induces a natural comparison test, $T_{cross}$,  defined as follows:
$T_{cross}=\frac{1}{2}$ if and only if both experts either pass or
fail the cross-calibration test whereas $T_{cross}=i$ if and only
if expert $i$ passes the cross-calibration test and expert $1-i$
fails. 
\begin{claim}
\label{claim 2 cross calibration is not reasonable} $T_{cross}$
is error-free and is not reasonable.
\end{claim}
\begin{proof}
Let $f_{0},f_{1}$ be forecasting strategies which deterministically
predict $\overset{0}{\omega}\coloneqq(0,1,1,,,),$ $\overset{1}{\omega,}$
respectively, and observe that since both $f_{0}$ and $f_{1}$ pass
the cross-calibration test on $h(\overset{0}{\omega},f_{0},f_{1})$
it follows that $T_{cross}(\stackrel{0}{\omega},f_{0},f_{1})=\frac{1}{2}$
yielding

\begin{equation}
1=P_{0}^{f}(\{\stackrel{0}{\omega}\})\leq P_{0}^{f}(\{T_{cross}(\cdot,f)=\frac{1}{2}\}).\label{eq: example CC-1-1-1-1-2-1}
\end{equation}

\noindent However, $f_{0},f_{1}$ are mutually singular with respect
to each other; so if $T_{cross}$ was a reasonable test then, by Lemma
\ref{lem:reasonable implies m.s}, it would satisfy

\noindent 
\[
P_{0}^{f}(\{T_{cross}(\cdot,f)=0\})=1
\]
which contradicts $(\ref{eq: example CC-1-1-1-1-2-1})$ and therefore
$T_{cross}$ is not reasonable. The fact that $T_{cross}$ is error-free
follows immediately from \Citet*{Dawid-1982} and hence omitted. 
\end{proof}
\vspace{0cm}

One could suspect that the counterexample used in the proof of Claim
\ref{claim 2 cross calibration is not reasonable} builds on the fact
that both experts use some Dirac measure and so assign zero probability
to any finite history that disagrees with that measure. Thus, a counterexample
where both forecasters assign a positive probability to any finite
history is provided in Appendix \ref{sec:Appendix B The-Cross-Calibration}.

\section{{\normalsize{}Summary}}

We study tests that compare two (self-proclaimed) experts in light
of some infinite sequence of forecasts and outcomes, where the goal
of the test is to spot the better informed one. We propose some natural
properties for such tests and construct the unique test (up to an
equivalence class) that complies with these properties. In \citet{Kavaler-Smorodinsky-2019}
we propose a framework where a comparison test provides a verdict
in finite time. We adapt the four properties to the new setting and
similarly propose a unique test for that environment. Some natural
directions for future research are to extend our results to settings
with more than two experts and to study alternative sets of properties. 

\vspace{0cm}

\section{{\normalsize{}Acknowledgments}}

The research of Smorodinsky is supported by the United States - Israel
Binational Science Foundation and the National Science Foundation
(grant 2016734), by the German-Israel Foundation (grant no. I-1419-118.4/2017),
by the Ministry of Science and Technology (grant 19400214) and by
Technion VPR grants and the Bernard M. Gordon Center for Systems Engineering
at the Technion. 

\vspace{0cm}

\bibliography{On_comparison_of_experts}

\section*{{\normalsize{} \center{APPENDIX}
}}

\appendix
%dummy comment inserted by tex2lyx to ensure that this paragraph is not empty

\section{{\normalsize{}\label{Appendix A sec:Missing-Proofs}Missing proofs}}

\setcounter{equation}{0} \renewcommand{\theequation}{A.\arabic{equation}}
\begin{lem}
\label{Lemma: pairwise disjoint cylinders}{}{}Let ${\cal B}\coloneqq\{B_{i}\}_{i\in\mathbb{N}}$
be an arbitrary sequence of cylinders and set $\overline{B}\coloneqq\underset{i\in\mathbb{N}}{\bigcup}B_{i}.$
Then, there exists an index set $J\subseteq\mathbb{N}$ such that
$\{B_{j}\}_{j\in J}$ are pairwise disjoint, and $\overline{B}=\underset{j\in J}{\bigcup}B_{j}$.
\end{lem}
\begin{proof}
\noindent A cylinder is called maximal in $\overline{B}$ if it is
not a subset of any other cylinders in ${\cal B}$. Any cylinder in
${\cal B}$ is contained in some maximal cylinder in $\overline{B}$.
Let $J\subseteq\mathbb{N}$ be such that $\{B_{j}\}_{j\in J}$ is
the set of all distinct maximal cylinders. Since any two distinct
maximal cylinders are disjoint it follows that $\overline{B}=\underset{j\in J}{\bigcup}B_{j}$. 
\end{proof}
\vspace{0cm}

\begin{proof}[\textbf{Proof of Lemma \ref{Lemma:liminf=00003Dlimsup} }]
{}{}(a) Let $A$ be a measurable set which satisfies the left side
of $(a)$ and let $U\subset\Omega^{\infty}$ be any open set such
that $A\subset U.$ Fix $\epsilon>0,$ then for all $a\in A,N>0$
there exists $t=t_{(a,N,\epsilon)}>N$ such that

\begin{equation}
D_{f_{0}}^{t}f_{1}(a)=\frac{\stackrel[n=1]{t}{\prod}f_{1}(h^{n-1}(a,f_{0},f_{1}))[a_{n}]}{\stackrel[n=1]{t}{\prod}f_{0}(h^{n-1}(a,f_{0},f_{1}))[a_{n}]}=\frac{P_{1}^{f}(a^{t})}{P_{0}^{f}(a^{t})}\leq(\alpha+\epsilon).\label{eq:10}
\end{equation}

Consider the following set of cylinders

\[
{\cal B}\coloneqq\{a^{t}\subset U|\ a\in A,\ t>0,\ P_{1}^{f}(a^{t})\leq(\alpha+\epsilon)P_{0}^{f}(a^{t})\}.
\]

\noindent Note, it follows from $\eqref{eq:10}$ that $\mathcal{{\cal B}}$
is not empty where $sup\{t|\ a^{t}\in{\cal B}\}=\infty.$ By Lemma
\ref{Lemma: pairwise disjoint cylinders} we are provided with an
index set $J$ and a collection of pairwise disjoint sets $\{B_{j}\in{\cal B}\}_{j\in J}$
such that

\noindent 
\begin{equation}
\overline{B}\coloneqq\underset{B\in{\cal B}}{\bigcup}B=\underset{j\in J}{\bigcup}B_{j}\label{eq:11}
\end{equation}

\noindent yielding that $A\subseteq\overline{B}$ and $B_{j}\in{\cal B}.$
Hence,

\[
\begin{array}{l}
P_{1}^{f}(A)\\
\\
\leq P_{1}^{f}(\overline{B})=P_{1}^{f}(\underset{j\in J}{\bigcup}B_{j})\leq\underset{j\in J}{\sum}P_{1}^{f}(B_{j})\leq\underset{j\in J}{\sum}(\alpha+\epsilon)P_{0}^{f}(B_{j})=(\alpha+\epsilon)\underset{j\in J}{\sum}P_{0}^{f}(B_{j})\leq(\alpha+\epsilon)P_{0}^{f}(U),
\end{array}
\]

\noindent where the most-right inequality holds since $U\supset B_{j}'s$
are disjoint.

Since the above inequalities hold for any open set $U$ which contains
$A$ and

\[
P_{0}^{f}(A)=\underset{U-open\colon A\subset U}{inf}\{P_{0}^{f}(U)\},
\]

\noindent it follows that for all $\epsilon>0,$

\[
P_{1}^{f}(A)\leq(\alpha+\epsilon)P_{0}^{f}(A)
\]

\noindent which completes the proof of Case (a). The proof of Case
(b) is analogous and hence omitted. 
\end{proof}
\noindent We now turn to show that the derivative of one measure with
respect to another exists and is finite almost surely. 

\vspace{0cm}

\begin{proof}[\textbf{Proof of Lemma \ref{Lemma: D exists and finite}}]
{}{}Let $I\coloneqq\{\omega|\ \overline{D}_{f_{0}}f_{1}(\omega)=+\infty\}$.
Therefore, for all $\alpha>0,$

\noindent 
\[
I\subset\{\omega|\ \overline{D}_{f_{0}}f_{1}(\omega)\geq\alpha\}
\]
and it follows from part b of Lemma \ref{Lemma:liminf=00003Dlimsup}
that $P_{0}^{f}(I)\leq\frac{1}{\alpha}P_{1}^{f}(I).$ Now let $\alpha\rightarrow\infty$
to obtain 
\begin{equation}
P_{0}^{f}(I)=0,\label{eq:12}
\end{equation}
and consequently $\overline{D}_{f_{0}}f_{1}$ is finite $P_{0}^{f}-a.e.$
For the second part let

\[
R(a,b)\coloneqq\{\omega|\ \underline{D}{}_{f_{0}}f_{1}(\omega)<a<b<\overline{D}_{f_{0}}f_{1}(\omega)<\infty\}.
\]

\noindent Note that

\[
R(a,b)\subset\{\omega|\ \underline{D}{}_{f_{0}}f_{1}(\omega)\leq a\}
\]

\noindent as well as

\[
R(a,b)\subset\{\omega|\ \overline{D}_{f_{0}}f_{1}(\omega)\geq b\}
\]

\noindent where applying Lemma \ref{Lemma:liminf=00003Dlimsup} in
both directions gives

\noindent 
\[
bP_{0}^{f}(R(a,b))\leq P_{1}^{f}(R(a,b))\leq aP_{0}^{f}(R(a,b)).
\]

Hence, for all $0<a<b$,
\begin{equation}
P_{0}^{f}(R(a,b))=0\label{eq:14}
\end{equation}
where from \eqref{eq:12} and \eqref{eq:14} we obtain

\noindent 
\[
\begin{array}{l}
P_{0}^{f}(\{\omega|\ \underline{D}{}_{f_{0}}f_{1}(\omega)<\overline{D}_{f_{0}}f_{1}(\omega)<\infty\})\\
\\
=P_{0}^{f}(\underset{\underset{a,b\in\mathbb{Q}}{0<a<b}}{\bigcup}R(a,b))\leq\underset{\underset{a,b\in\mathbb{Q}}{0<a<b}}{\sum}P_{0}^{f}(R(a,b))=0.
\end{array}
\]
Therefore, $D_{f_{1}}f_{0}$ exists $P_{0}^{f}$- a.e. 
\end{proof}
\vspace{0cm}

\begin{proof}[\textbf{Proof of Proposition \ref{prop:equivalence}}]
Let $T,T_{1},T_{2}\in\top,\ f\in F,\text{ and }i\in\{0,1\}.$

\noindent Reflexivity:

\[
P_{i}^{f}(\{\omega|\;T(\omega,f_{0},f_{1})\neq T(\omega,f_{0},f_{1})\})=0\Longrightarrow T\sim T.
\]

\noindent Symmetry:

\[
P_{i}^{f}(\{\omega|\;T_{1}(\omega,f_{0},f_{1})\neq T_{2}(\omega,f_{0},f_{1})\})=0\iff P_{i}^{f}(\{\omega|\;T_{2}(\omega,f_{0},f_{1})\neq T_{1}(\omega,f_{0},f_{1})\})=0;
\]

\noindent hence, $T_{1}\sim T_{2}\iff T_{2}\sim T_{1}.$

\noindent Transitivity: Assume that $T_{1}\sim T,$ and $T\sim T_{2};$
hence

\[
T_{1}\sim_{f}T\Longrightarrow P_{i}^{f}(\{\omega|\;T_{1}(\omega,f_{0},f_{1})\neq T(\omega,f_{0},f_{1})\}^{c})=1,
\]

\noindent as well as

\[
T\sim_{f}T_{2}\Longrightarrow P_{i}^{f}(\{\omega|\;T(\omega,f_{0},f_{1})\neq T_{2}(\omega,f_{0},f_{1})\}^{c})=1.
\]

\noindent Thus

\[
\begin{array}{l}
P_{i}^{f}(\{\omega|\;T_{1}(\omega,f_{0},f_{1})\neq T_{2}(\omega,f_{0},f_{1})\}^{c})\\
\\
=P_{i}^{f}(\{\omega|\;T_{1}(\omega,f_{0},f_{1})\neq T(\omega,f_{0},f_{1})\}^{c}\cap\{\omega|\;T(\omega,f_{0},f_{1})\neq T_{2}(\omega,f_{0},f_{1})\}^{c})=1,
\end{array}
\]

\noindent yielding $P_{i}^{f}(\{\omega|\;T_{1}(\omega,f_{0},f_{1})\neq T_{2}(\omega,f_{0},f_{1})\})=0,$
and therefore $T_{1}\sim_{f}T_{2}.$ 
\end{proof}
\vspace{0cm}

\begin{proof}[\textbf{Proof of Lemma \ref{lem:reasonable implies m.s}}]
W.l.o.g.$\ $let $A$ be such that: $P_{0}^{f}(A)=1,P_{1}^{f}(A)=0$.
$T$ is reasonable, therefore $P_{0}^{f}(A\cap\{T(\cdot,f)=0\})>0$
from \eqref{eq:reasonableness condition}. Let $k\in\{\frac{1}{2},1\}$
and assume that

\[
P_{0}^{f}(A\cap\{T(\cdot,f)=k\})>0.
\]

\noindent Lemma \ref{Lemma  if T reasonable then P1(A intersect AT0)}
 yields

\noindent 
\[
P_{1}^{f}(A\cap\{T(\cdot,f)=k\})>0
\]
which contradicts the assumption that $P_{1}^{f}(A)=0.$ Hence, $P_{0}^{f}(A\cap\{T(\cdot,f)=k\})=0.$
As a result,

\[
P_{0}^{f}(A\cap\{T(\cdot,f)=0\})=P_{0}^{f}(A)=1
\]

\noindent and therefore $P_{0}^{f}(\{T(\cdot,f)=0\})=1.$ 

\vspace{0cm}
\end{proof}

\section{{\normalsize{}\label{sec:Appendix B The-Cross-Calibration}The cross-calibration
test}}

We now restate the cross-calibration test as suggested by \Citet*{Feinberg-Stewart-2008}.
Fix a positive integer $N>4$ and divide the interval $[0,1]$ into
$N$ equal closed subintervals $I_{1},...,I_{N},$ so that $I_{j}=[\frac{j-1}{N},\frac{j}{N}],\ 1\leq j\leq N$.
All results in their paper hold when $[0,1]$ is replaced with the
set of distributions over any finite set and the intervals $I_{j}$
are replaced with a cover of the set of distributions by sufficiently
small closed convex subsets. At the beginning of each period $t=1,2...,$
all forecasters (or experts) $i\in\{0,..,M-1\}$ simultaneously announce
predictions $I_{t}^{i}\in\{I_{1},...,I_{N}\}$, which are interpreted
as probabilities with which the outcome $1$ will occur in that period.
We assume that forecasters observe both the realized outcome and the
predictions of the other forecasters at the end of each period. 

The cross-calibration test is defined over sequences $(\omega_{t},I_{t}^{0},...,I_{t}^{M-1})_{t=1}^{\infty}$,
which specify, for each period $t$, the outcome $\omega_{t}\in\Omega$,
together with the prediction intervals announced by each of the $M$
forecasters. Given any such sequence and any $M$ - tuple $l=(I_{l^{0}},...,I_{l^{M-1}})\in\{I_{1},...,I_{N}\}^{M}$,
define $\zeta_{t}^{l}=1_{I_{t}^{i}=I_{l^{i}},\forall i=0,...,M-1},$ and $\nu_{t}^{l}=\stackrel[n=1]{t}{\sum}\zeta_{n}^{l},$  where  $\nu_{t}^{l}$ represents the number of times that the forecast profile
$l$ is chosen up to time $t$. For $\nu_{t}^{l}>0$, the frequency
$f_{t}^{l}$ of outcomes conditional on this forecast profile is given
by

\[
f_{t}^{l}=\frac{1}{\nu_{t}^{l}}\stackrel[n=1]{t}{\sum}\zeta_{n}^{l}\omega_{n}.
\]

Forecaster $i$ passes the cross-calibration test at the sequence
$(\omega_{t},I_{t}^{0},...,I_{t}^{M-1})_{t=1}^{\infty}$ if

\begin{equation}
\underset{t\rightarrow\infty}{limsup}|f_{t}^{l}-\frac{2l^{i}-1}{2N}|\leq\frac{1}{2N}\label{eq:passing the cross calibration test}
\end{equation}

\noindent for every $l$ satisfying $\underset{t\rightarrow\infty}{lim}\nu_{t}^{l}=\infty$. 

In the case of a single forecaster, the cross-calibration test reduces
to the classic calibration test, which checks the frequency of outcomes
conditional on each forecast that is made infinitely often. With multiple
forecasters, the cross-calibration test checks the empirical frequencies
of the realization conditional on each profile of forecasts that occurs
infinitely often. Note that if an expert is cross-calibrated, he will
also be calibrated.

\vspace{0cm}

Claim \ref{claim 2 cross calibration is not reasonable} demonstrated
why $T_{cross}$ is not reasonable, and so does not satisfy the set
of axioms we study. In that example, both forecasters used a Dirac
measure. We now turn to a slightly more elaborate example that demonstrates
the same thing; yet the forecasters assign a positive probability
to any finite history.
\begin{example}
\label{example: cross calibration positive each history}Set $N>4,\ M=2.$
Let $f_{0}$ be a convex combination of two forecasting strategies.
With probability $0.5$ it deterministically predicts $\overset{1}{\omega}$
and with the remaining probability it is an IID sequence of fair coin
flips. On the other hand, $f_{1}$ forecasts $1$ in period $t$ with
probability $1-\frac{1}{(t+2)}$, independent of past outcomes.

Then, conditional on the realization of $\stackrel{1}{\omega},$ both
experts repeatedly announce the interval $I_{N}$ from some finite
time onward. Consequently, over the profile $l=(1,I_{N},I_{N}),$
equation \eqref{eq:passing the cross calibration test} holds for
all $i$ and therefore both experts pass the cross-calibration test
over $\overset{1}{\omega}$ yielding that

\begin{equation}
T_{cross}(\stackrel{1}{\omega},f_{0},f_{1})=\frac{1}{2}.\label{eq:T_cross=00003D0.5-1}
\end{equation}

However, by construction, $P_{0}^{f}(\{\overset{1}{\omega}\})=\frac{1}{2}$
and $P_{1}^{f}(\{\overset{1}{\omega}\})=0,$ and yet, if $T_{cross}$
would be a reasonable test, then $0<P_{0}^{f}(\{\stackrel{1}{\omega}\})=P_{0}^{f}(\{T_{cross}(\cdot,f)=0\}\cap\{\stackrel{1}{\omega}\})$
which contradicts equality \eqref{eq:T_cross=00003D0.5-1}. 
\end{example}
\vspace{0cm}

\end{document}